\documentclass[preprint,
showpacs,
amsmath,amssymb]{revtex4}

\usepackage{graphicx}
\usepackage{subfigure}
\usepackage{amsthm}
\usepackage{color}
\usepackage[all]{xy}
\usepackage{tikz}
\usepackage{dsfont}
\usetikzlibrary{positioning}

\newtheorem{Thm}{Theorem}
\newtheorem{Lem}[Thm]{Lemma}

\newtheorem{Cor}[Thm]{Corollary}
\theoremstyle{definition}
\newtheorem{Def}[Thm]{Definition}

\newcommand{\ket}[1]{\left|{#1}\right\rangle}
\newcommand{\bra}[1]{\left\langle{#1}\right|}

\begin{document} 

\title{State transfer with quantum side information}

\author{Yonghae Lee}
\author{Soojoon Lee}
\affiliation{
 Department of Mathematics and Research Institute for Basic Sciences,
 Kyung Hee University, Seoul 02447, Korea
}

\date{\today}

\begin{abstract}
We first consider quantum communication protocols 
between a sender Alice and a receiver Bob, 
which transfer Alice's quantum information to Bob 
by means of non-local resources, 
such as classical communication, quantum communication, and entanglement. 
In these protocols, 
we assume that Alice and Bob may have quantum side information, 
not transferred. 
In this work, 
these protocols are called the state transfer with quantum side information.
We determine the optimal costs for non-local resources in the protocols, 
and study what the effects of the use of quantum side information are.
Our results can give new operational meanings 
to the quantum mutual information and the quantum conditional mutual information, 
which directly provide us with an operational interpretation of the chain rule 
for the quantum mutual information.
\end{abstract}

\maketitle
\section{Introduction}

There are quantum communication protocols,
such as the quantum teleportation~\cite{BBCJPW93} and the Schumacher compression~\cite{S95},
which transfer quantum information from Alice to Bob.
In quantum information theory,
these protocols have been regarded as the leading research topics,
since they can provide new operational meanings
to quantum quantities,
such as the von Neumann entropies~\cite{W13} and the smooth entropies~\cite{KRS09}. 
New operational meanings have made the quantum information
theory richer through intuitive understandings
of quantum phenomena. 

We here consider protocols in which Alice's information can be asymptotically transferred to Bob 
by means of quantum/classical communication and entanglement as non-local resources.
In the protocols,
Alice and Bob are able to apply local operations on their states,
and employ their quantum side information (QSI)
in order to transfer Alice's information.
We call the protocols the {\em state transfer with QSI},
and divide the state transfer protocols with QSI into two types:
the {\em state redistribution with QSI} and the {\em state merging with QSI}.
In the former Alice and Bob use quantum channels for communication from Alice to Bob,
and in the latter they use classical channels.

Although there have been some protocols~\cite{HOW05,HOW07,DY08,YD09,D06,ADHW09,F04,LKPL09,O08}
which deal with Alice's or Bob's QSI,
the results have not explicitly explained
how the use of QSI has the effects on the optimal resource costs.
In addition,
when Alice and Bob can use more (or less) QSI,
it has not been mentioned in literature.
On this account,
one can raise the following two questions:
(i) How does the use of Alice's and Bob's QSI affect the optimal resource costs
in the state transfer with QSI?
(ii) Assume that Alice or Bob uses more (or less) QSI
in the state transfer with QSI.
How does the use of more (or less) QSI affect the optimal resource costs?

In order to answer the two questions,
we describe a mathematical definition of the state transfer with QSI,
and calculate its optimal costs for non-local resources.
Then we study the effects of QSI on the optimal resource costs
of the state transfer with QSI.
From these results,
we present new operational meanings 
of the quantum mutual information (QMI),
quantum conditional mutual information (QCMI), 
and a new operational interpretation of the chain rule
for the QMI~\cite{W13}.

This paper is organized as follows.
In Sec.~\ref{sec:STwQSI}
we define the state transfer with QSI,
and calculate its optimal costs for non-local resources. 
In Sec.~\ref{sec:Effects_of_QSI}
we study what the effects of the use of QSI are
in the state transfer with QSI.
Then we give new operational meanings
to the QMI and the QCMI in Sec.~\ref{sec:New_operational_meanings}.
We also present well-known examples
which are special cases of the state transfer with QSI
in Sec.~\ref{sec:Examples_of_STwQSI}.
Finally, in Sec.~\ref{sec:Conclusion}
we summarize and discuss our results.

\section{State transfer with QSI} \label{sec:STwQSI}
We formally define the state transfer with QSI as follows.

\begin{Def}[State transfer with QSI] \label{def:STwQSI}
Let $\ket{\psi}\equiv\ket{\psi}_{A_{1}\cdots A_{m}C_{\mathrm{A}}B_{1}\cdots B_{n}R}$
be a pure initial state,
where Alice and Bob hold $A_{1}\cdots A_{m}C_{\mathrm{A}}$ and ${B_{1}\cdots B_{n}}$,
respectively,
and $R$ is the reference.
Assume that Alice and Bob have additional systems
$E_{\mathrm{A}}^{\mathrm{in}}$, $E_{\mathrm{A}}^{\mathrm{out}}$ 
and $E_{\mathrm{B}}^{\mathrm{in}}$, $E_{\mathrm{B}}^{\mathrm{out}}$ for entanglement resources,
respectively.
For $0\le i\le m$ and $0\le j\le n$,
a joint operation
\begin{eqnarray} 
\mathcal{T}_{ij}&:&A_{1}\cdots A_{i}C_{\mathrm{A}}E_{\mathrm{A}}^{\mathrm{in}}\otimes
B_{1}\cdots B_{j}E_{\mathrm{B}}^{\mathrm{in}} \nonumber \\
&&\longrightarrow A_{1}\cdots A_{i}E_{\mathrm{A}}^{\mathrm{out}}\otimes
C_{\mathrm{B}}B_{1}\cdots B_{j}E_{\mathrm{B}}^{\mathrm{out}} \nonumber
\end{eqnarray}
is called the \emph{state transfer with QSI of $\ket{\psi}$}
(or $\mathrm{tr}_{R}(\ket{\psi}\bra{\psi})$ \emph{with error $\varepsilon$},
if it consists of local operations and either qubit channels or bit channels from Alice to Bob,
and satisfies
\begin{eqnarray}
&F\Big(
(\mathcal{T}_{ij}\otimes\mathds{1}_{A_{i+1}\cdots A_{m}B_{j+1}\cdots B_{n}R})
\big(\ket{\psi}
\otimes\ket{\Phi}_{E_{\mathrm{A}}^{\mathrm{in}}E_{\mathrm{B}}^{\mathrm{in}}}\big), \nonumber \\
&\ket{\psi'}\otimes\ket{\Phi}_{E_{\mathrm{A}}^{\mathrm{out}}E_{\mathrm{B}}^{\mathrm{out}}} 
\Big)\ge1-\varepsilon, \nonumber
\end{eqnarray}
where 
$C_{\mathrm{B}}$ is Bob's system with $\dim C_{\mathrm{B}}=\dim C_{\mathrm{A}}$,
$F(\cdot,\cdot)$ is the quantum fidelity,
$\ket{\psi'}$ is a final state defined as 
$\left(\mathds{1}_{A_{1}\cdots A_{m}B_{1}\cdots B_{n}R}
\otimes\mathds{1}_{C_{\mathrm{A}}\to C_{\mathrm{B}}}\right)\ket{\psi}$, 
and $\ket{\Phi}_{E_{\mathrm{A}}^{\mathrm{in}}E_{\mathrm{B}}^{\mathrm{in}}}$
and $\ket{\Phi}_{E_{\mathrm{A}}^{\mathrm{out}}E_{\mathrm{B}}^{\mathrm{out}}}$
are maximally entangled states 
with Schmidt-rank $e^{\mathrm{in}}(\mathcal{T}_{ij})$ and $e^{\mathrm{out}}(\mathcal{T}_{ij})$, 
respectively.

In addition, 
we call the operation $\mathcal{T}_{ij}$ the \emph{state redistribution with QSI},
if it consists of local operations and $q(\mathcal{T}_{ij})$ qubit channels
without any classical channels,
and $\mathcal{T}_{ij}$ is called the \emph{state merging with QSI},
if it consists of local operations and $c(\mathcal{T}_{ij})$ bit channels
without any quantum channels.
\end{Def}

\begin{figure}
\centering
\includegraphics[width=.9\linewidth,trim=0cm 0cm 0cm 0cm]{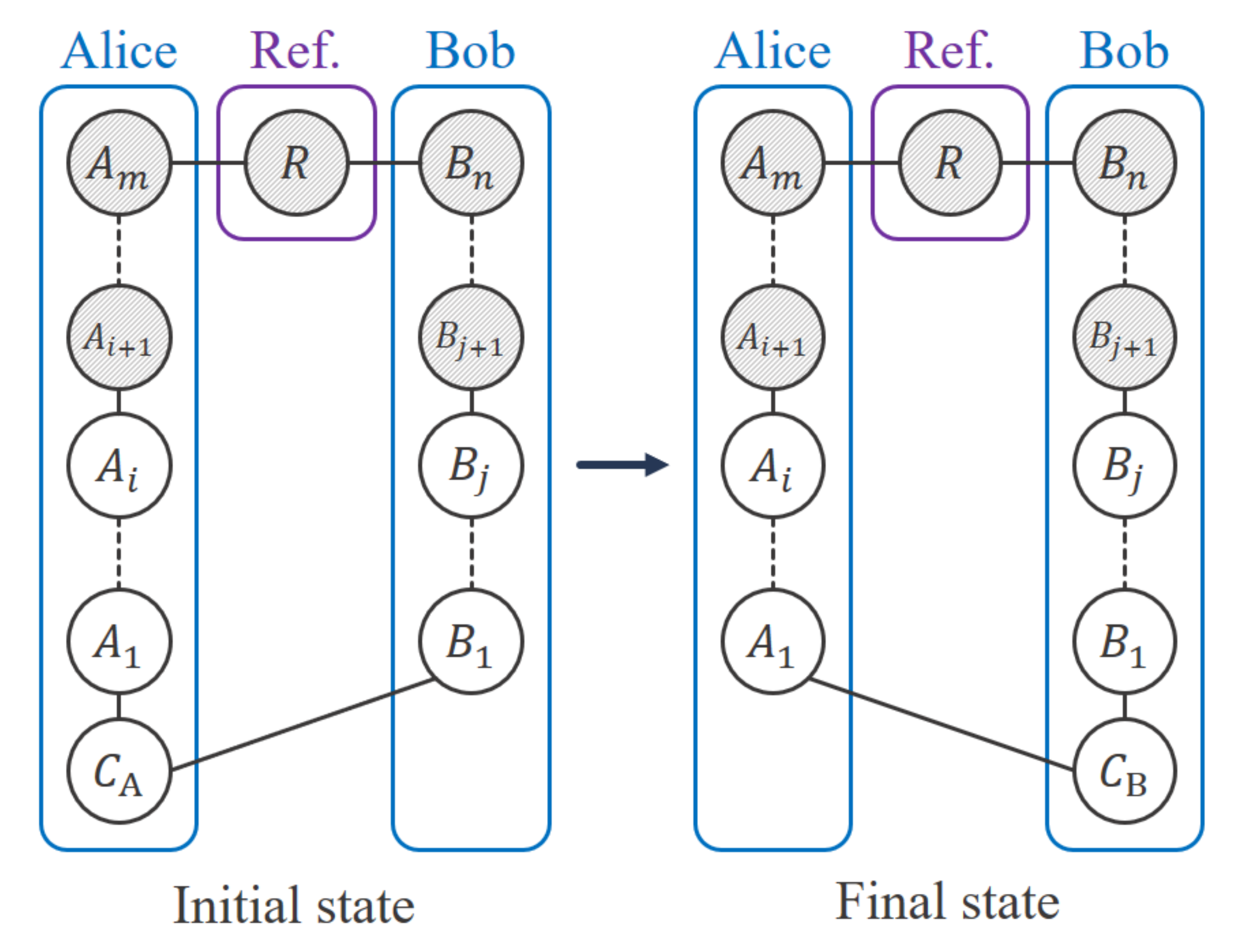}
\caption{
The initial and final states
for the state transfer with QSI of $\ket{\psi}$: $C_{\mathrm{A}}$ is a transferred part,
$A_{1}\cdots A_{m}$ and $B_{1}\cdots B_{n}$ are Alice's and Bob's QSI, and $R$ is the reference.
In the state transfer with QSI,
they use QSI $A_{1}\cdots A_{i}$ and $B_{1}\cdots B_{j}$, respectively,
while the rests $A_{i+1}\cdots A_{m}$ and $B_{j+1}\cdots B_{n}$ are left unused.
}
\label{fig:Protocol}
\end{figure}

In Definition~\ref{def:STwQSI}, 
the indices $i$ and $j$ of $\mathcal{T}_{ij}$ mean 
that Alice and Bob apply local operations on their QSI $A_{1}\cdots A_{i}$ and $B_{1}\cdots B_{j}$
in order to transfer Alice's $C_{\mathrm{A}}$ to Bob as depicted in Fig.~\ref{fig:Protocol},
and in this situation
we say that Alice and Bob \emph{use} their QSI $A_{1}\cdots A_{i}$ and $B_{1}\cdots B_{j}$.
For instance, Alice and Bob do not use any QSI if and only if $i=0$ and $j=0$, respectively, 
and they make use of the whole QSI if and only if $i=m$ and $j=n$, respectively.

We also define the optimal resource costs of the state transfer with QSI of $\ket{\psi}$
for fixed $i$ and $j$.

\begin{Def}
For $n$ independent and identically distributed copies
of $\ket{\psi}\equiv\ket{\psi}_{A_{1}\cdots A_{m}C_{\mathrm{A}}B_{1}\cdots B_{n}R}$,
say $\ket{\psi}^{\otimes n}$,
let $\mathcal{T}_{ij}^n$ be a state redistribution (or a state merging)
with QSI of $\ket{\psi}^{\otimes n}$ with error $\varepsilon_n$, 
then the resource rates
$(\log e^{\mathrm{in}}(\mathcal{T}_{ij}^n)-\log e^{\mathrm{out}}(\mathcal{T}_{ij}^n))/n$
and $q(\mathcal{T}_{ij}^n)/n$ 
(or $c(\mathcal{T}_{ij}^n)/n$) are called the \emph{entanglement rate}
and \emph{quantum communication rate} 
(or \emph{classical communication rate}) of the protocol,
respectively.

For each resource rate,
we call a real number $r$ an \emph{achievable rate}
if there is a sequence $\{\mathcal{T}_{ij}^n\}_{n\in\mathbb{N}}$
such that the sequence $\{\varepsilon_n\}_{n\in\mathbb{N}}$ converges to zero,
and the sequence for the resource rate converges $r$ 
as $n$ tends to infinity.
The smallest achievable rates for entanglement and quantum communication
(or classical communication) are called
the \emph{optimal entanglement cost} and \emph{optimal quantum communication cost}
(or \emph{optimal classical communication cost}),
respectively.
\end{Def}

We investigate the optimal resource costs for the state redistribution with QSI
of $\ket{\psi}\equiv\ket{\psi}_{A_{1}\cdots A_{m}C_{\mathrm{A}}B_{1}\cdots B_{n}R}$.
Let $Q_{i,j}$ and $E_{i,j}$ be its optimal quantum communication and entanglement costs,
respectively,
when Alice and Bob use QSI $A_{1}\cdots A_{i}$ and $B_{1}\cdots B_{j}$.
Let $\tilde{A}=A_{1}\cdots A_{i}$, $\tilde{B}=B_{1}\cdots B_{j}$,
and $\tilde{R}=A_{i+1}\cdots A_{m}B_{j+1}\cdots B_{n}R$.
Then the given state $\ket{\psi}$ becomes a four-partite state
$\ket{\psi}_{\tilde{A}C_{\mathrm{A}}\tilde{B}\tilde{R}}$.
Since $A_{i+1}\cdots A_{m}$ and $B_{j+1}\cdots B_{n}$ are not used,
and $\tilde{R}$ can be considered as the reference system of a purification $\ket{\psi}$ 
of a quantum state $\rho_{\tilde{A}C_{\mathrm{A}}\tilde{B}}$,
our state redistribution with QSI is identical to the state redistribution 
for $\ket{\psi}_{\tilde{A}C_{\mathrm{A}}\tilde{B}\tilde{R}}$~\cite{DY08,YD09}.
Thus, 
we can obtain that
\begin{eqnarray}
Q_{i,j}&=&\frac{1}{2}I(C_{\mathrm{A}};\tilde{R}|\tilde{B}) \nonumber \\
       &=&H(C_{\mathrm{A}})-\frac{1}{2}I(C_{\mathrm{A}};\tilde{A})
                           -\frac{1}{2}I(C_{\mathrm{A}};\tilde{B}), \nonumber \\
E_{i,j}&=&\frac{1}{2}I(C_{\mathrm{A}};\tilde{A})-\frac{1}{2}I(C_{\mathrm{A}};\tilde{B}), \nonumber
\end{eqnarray} 
where $I(\cdot;\cdot|\cdot)$ is the QCMI,
$H(\cdot)$ is the von Neumann entropy and $I(\cdot;\cdot)$ is the QMI.
This implies the following lemma.

\begin{Lem} \label{lem:Opt_Cost}
For a state $\rho_{A_{1}\cdots A_{m}C_{\mathrm{A}}B_{1}\cdots B_{n}}$ shared by Alice and Bob, 
the optimal quantum communication cost $Q_{i,j}$ 
and the optimal entanglement cost $E_{i,j}$ for the state redistribution with QSI
can be expressed as the von Neumann entropy $H(C_{\mathrm{A}})$ 
and the QMI $I(C_{\mathrm{A}};A_{1}\cdots A_{i})$ and $I(C_{\mathrm{A}};B_{1}\cdots B_{j})$
as follows: 
\begin{eqnarray} \label{eq:Opt_Cost}
Q_{i,j}&=&H(C_{\mathrm{A}})-\frac{1}{2}I(C_{\mathrm{A}};A_{1}\cdots A_{i})
                           -\frac{1}{2}I(C_{\mathrm{A}};B_{1}\cdots B_{j}), \nonumber \\
E_{i,j}&=&\frac{1}{2}I(C_{\mathrm{A}};A_{1}\cdots A_{i})
         -\frac{1}{2}I(C_{\mathrm{A}};B_{1}\cdots B_{j}).
\end{eqnarray}
\end{Lem}

By replacing qubit channels with bit channels,
we can consider the state merging with QSI of the state $\ket{\psi}$.
For each $0\le i\le m$ and $0\le j\le n$, 
let $c_{i,j}$ and $e_{i,j}$ be 
the optimal classical communication and entanglement costs of the state merging with QSI,
respectively,
when Alice and Bob employ QSI $A_{1}\cdots A_{i}$ and $B_{1}\cdots B_{j}$.
Then we obtain the following lemma.

\begin{Lem} \label{lem:optimal_c_costs}
For each $0\le i\le m$ and $0\le j\le n$,
the optimal classical communication cost $c_{i,j}$ and the optimal entanglement cost $e_{i,j}$
for the state merging with QSI can be expressed
in terms of the optimal costs $Q_{i,j}$ and $E_{i,j}$ for the state redistribution with QSI
as follows:
\begin{eqnarray}
c_{i,j}&=&2 Q_{i,j}, \nonumber \\
e_{i,j}&=&Q_{i,j}+ E_{i,j}. \nonumber
\end{eqnarray} 
\end{Lem}

\begin{proof}
We first note that $Q_{i,j}$ qubit channels can be perfectly simulated
with $2Q_{i,j}$ bit channels and $Q_{i,j}$ ebits
by the quantum teleportation~\cite{BBCJPW93}.
Thus by Lemma~\ref{lem:Opt_Cost} Alice and Bob can perform the state merging with QSI
by consuming $2Q_{i,j}$ bit channels and $Q_{i,j}+E_{i,j}$ ebits.

Now,
we show that the costs of $2Q_{i,j}$ bit channels and $Q_{i,j}+E_{i,j}$ ebits are optimal
for the state merging with QSI.

Suppose that the cost of $2Q_{i,j}$ bit channels is not optimal,
that is, 
there exists $c_{i,j}'$ such that $c_{i,j}'<2Q_{i,j}$
and the state merging with QSI can be performed with $c_{i,j}'$ bit channels. 
Then as in the proof of the optimality for the classical communication cost
in the state merging~\cite{HOW07},
$c_{i,j}'$ bit channels can be replaced by $c_{i,j}'/2$ qubit channels and $-c_{i,j}'/2$ ebits
through the coherent bit channel~\cite{DHW04,H04}.
Thus, 
the state redistribution with QSI can be performed with $c_{i,j}'/2$ qubit channels,
which contradicts the optimality of the quantum communication cost
for the state redistribution with QSI in Lemma~\ref{lem:Opt_Cost}.
Therefore, 
the optimal classical communication cost is $2Q_{i,j}$.

Finally, 
suppose that there exists $e_{i,j}'$ such that $e_{i,j}'<Q_{i,j}+E_{i,j}$
and the state merging with QSI can be performed with $e_{i,j}'$ ebits and $2Q_{i,j}$ bit channels.
Since $2Q_{i,j}$ bit channels can be replaced by $Q_{i,j}$ qubit channels and $-Q_{i,j}$ ebits,
it is possible to perform the state redistribution with QSI with $e_{i,j}'-Q_{i,j}$ ebits.
This contradicts the optimality of the entanglement cost for the state redistribution with QSI
in Lemma~\ref{lem:Opt_Cost}.
Therefore, 
the optimal entanglement cost is $Q_{i,j}+E_{i,j}$. 
\end{proof}

By Lemma~\ref{lem:Opt_Cost} and Lemma~\ref{lem:optimal_c_costs}, 
we can obtain that the optimal costs $c_{i,j}$ and $e_{i,j}$ of the state merging with QSI become
\begin{eqnarray} \label{eq:Opt_Cost_SM}
c_{i,j}&=&2H(C_{\mathrm{A}})-I(C_{\mathrm{A}};A_{1}\cdots A_{i})
                            -I(C_{\mathrm{A}};B_{1}\cdots B_{j}), \nonumber \\
e_{i,j}&=& H(C_{\mathrm{A}})-I(C_{\mathrm{A}};B_{1}\cdots B_{j})
\end{eqnarray}
for $0\le i\le m$ and $0\le j\le n$.

\section{Effects of QSI on optimal resource costs in State transfer with QSI}
\label{sec:Effects_of_QSI}
In this section,
we investigate how the use of (additional) QSI affects the optimal resource costs
in the state transfer with QSI.
For this,
we consider the state transfer with QSI
of $\rho\equiv\rho_{A_{1}\cdots A_{m}C_{\mathrm{A}}B_{1}\cdots B_{n}}$
which is shared by Alice and Bob as in Definition~\ref{def:STwQSI}.
In this state transfer with QSI of $\rho$,
$O$ denotes a type of non-local resources.
For instance,
$O$ can present one of non-local resources $Q$, $E$, $c$, or $e$.
Here,
$Q$ and $c$ are qubit channels and bit channels consumed
in the state transfer with QSI,
respectively.
$E$ ($e$) is ebits consumed/generated
in the state redistribution with QSI
(in the state merging with QSI).
For $0\le i\le m$ and $0\le j\le n$,
if Alice and Bob use QSI $A_{1}\cdots A_{i}$ and $B_{1}\cdots B_{j}$
in the state transfer with QSI of $\rho$
then the following definition enables us to quantify
the effects of their QSI on an optimal resource cost $R$
in the state transfer with QSI of $\rho$.

\begin{Def} \label{def:effect_of_QSI}
Let $\textbf{E}[O]_{i,j}=O_{0,0}-O_{i,j}$.
Then $\textbf{E}[O]_{i,j}$ is called 
the \emph{effect on the optimal resource cost of type $O$
with respect to QSI $A_{1}\cdots A_{i}$ and $B_{1}\cdots B_{j}$}
in the state transfer with QSI of $\rho$.
\end{Def}

The effect $\textbf{E}[O]_{i,j}$ in Definition~\ref{def:effect_of_QSI}
appropriately measures the effect of QSI $A_{1}\cdots A_{i}$ and $B_{1}\cdots B_{j}$
in the state transfer with QSI of $\rho$,
since the only difference
between the optimal resource costs $R_{0,0}$ and $R_{i,j}$
is the use of QSI $A_{1}\cdots A_{i}$ and $B_{1}\cdots B_{j}$.

From the formulas for the optimal costs in Eqs.~(\ref{eq:Opt_Cost}) and~(\ref{eq:Opt_Cost_SM}),
the effect $\textbf{E}[O]_{i,j}$ on the optimal resource cost of type $O$
is readily calculated.
Specifically, for the state redistribution with QSI of $\rho$,
the effects $\textbf{E}[Q]_{i,j}$ and $\textbf{E}[E]_{i,j}$
on the optimal quantum communication cost
and the optimal entanglement cost are given by 
\begin{eqnarray} \label{eq:effects_SRwQSI}
\textbf{E}[Q]_{i,j}&=& \frac{1}{2}I(C_{\mathrm{A}};A_{1}\cdots A_{i})
                       +\frac{1}{2}I(C_{\mathrm{A}};B_{1}\cdots B_{j}), \nonumber \\
\textbf{E}[E]_{i,j}&=&-\frac{1}{2}I(C_{\mathrm{A}};A_{1}\cdots A_{i})
                       +\frac{1}{2}I(C_{\mathrm{A}};B_{1}\cdots B_{j}).
\end{eqnarray}
For the state merging with QSI of $\rho$,
the effects $\textbf{E}[c]_{i,j}$ and $\textbf{E}[e]_{i,j}$
on the optimal classical communication cost
and the optimal entanglement cost are 
\begin{eqnarray} \label{eq:effects_SMwQSI}
\textbf{E}[c]_{i,j}&=& I(C_{\mathrm{A}};A_{1}\cdots A_{i})
                       +I(C_{\mathrm{A}};B_{1}\cdots B_{j}), \nonumber \\
\textbf{E}[e]_{i,j}&=& I(C_{\mathrm{A}};B_{1}\cdots B_{j}).
\end{eqnarray}

From Eqs.~(\ref{eq:effects_SRwQSI}) and~(\ref{eq:effects_SMwQSI}), 
it is observed that
the effects of QSI $A_{1}\cdots A_{i}$ and $B_{1}\cdots B_{j}$
can be decomposed
according to Alice's QSI $A_{1}\cdots A_{i}$ and Bob's QSI $B_{1}\cdots B_{j}$.
This means that
the use of Alice's QSI $A_{1}\cdots A_{i}$ and the use of Bob's QSI $B_{1}\cdots B_{j}$
independently affect the optimal resource costs in the state transfer with QSI of $\rho$.
The second observation is
that all effects of QSI stem from
the correlation between the part $C_{\mathrm{A}}$ and QSI $A_{1}\cdots A_{i}$
(or $B_{1}\cdots B_{j}$).

From these observations,
it follows that
the effect $\textbf{E}[O]_{i,j}$ on the optimal resource cost of type $O$ can be decomposed as
\begin{equation}
\textbf{E}[O]_{i,j}=\textbf{A}[O]_{i}+\textbf{B}[O]_{j}, \nonumber
\end{equation}
where $\textbf{A}[O]_{i}=\textbf{E}[O]_{i,0}$
and $\textbf{B}[O]_{j}=\textbf{E}[O]_{0,j}$.
Here,
$\textbf{A}[O]_{i}$ ($\textbf{B}[O]_{j}$) indicates
the effect of Alice's QSI $A_{1}\cdots A_{i}$ (Bob's QSI $B_{1}\cdots B_{j}$)
on the optimal resource cost of type $O$
for the state transfer with QSI of $\rho$.
This leads us to the following theorem
which provides answers about the first question.  

\begin{Thm} \label{thm:effect_of_QSI}
In the state transfer with QSI of $\rho$,
the effects of Alice's QSI $A_{1}\cdots A_{i}$ are simply expressed as $\textbf{A}[e]_{i}=0$ and
\begin{equation}
\textbf{A}[c]_{i}=2\textbf{A}[Q]_{i}=-2\textbf{A}[E]_{i}
=I(C_{\mathrm{A}};A_{1}\cdots A_{i}). \nonumber
\end{equation}
For the case of Bob's QSI $B_{1}\cdots B_{j}$,
the effects are
\begin{equation}
\textbf{B}[c]_{j}=\textbf{B}[e]_{j}=2\textbf{B}[Q]_{j}=2\textbf{B}[E]_{j}
=I(C_{\mathrm{A}};B_{1}\cdots B_{j}). \nonumber
\end{equation}
\end{Thm}

It is worth mentioning that
since the QMI is always non-negative,
the use of Bob's QSI $B_{1}\cdots B_{j}$ can reduce all optimal resource costs
of the state transfer with QSI
compared to the case that Bob uses no QSI.
On the other hand,
the effects of Alice's QSI are somewhat different.
If Alice uses her QSI $A_{1}\cdots A_{i}$,
then the optimal quantum/classical communication costs can be reduced,
since the effects $\textbf{A}[Q]_{i}$ and $\textbf{A}[c]_{i}$ are non-negative.
However,
from the fact that $\textbf{A}[e]_{i}=0$ and $\textbf{A}[E]_{i}$ is non-positive,
the optimal entanglement cost for the state merging with QSI is unchanged
and that for the state redistribution with QSI can increase.
This means that
even if Alice's QSI is sufficiently large,
the use of the QSI cannot reduce the optimal entanglement cost
of the state transfer with QSI,
and can even increase that of the state merging with QSI.

In order to answer the second question about additional QSI,
we need to consider the state transfer with QSI of $\rho$
which is shared by Alice and Bob as before.
Let $0\le i_1\le i_2\le m$ and $0\le j_1\le j_2\le n$.
In this state transfer with QSI of $\rho$,
Alice and Bob first use QSI $A_{1}\cdots A_{i_1}$ and $B_{1}\cdots B_{j_1}$.
Then they use more QSI $A_{1}\cdots A_{i_2}$ and $B_{1}\cdots B_{j_2}$,
so that QSI $A_{i_1+1}\cdots A_{i_2}$ and $B_{j_1+1}\cdots B_{j_2}$ are additionally used
in this situation.

We define the effects of the use
of the additional QSI $A_{i_1+1}\cdots A_{i_2}$ and $B_{j_1+1}\cdots B_{j_2}$
on the optimal resource cost of type $O$
in the state transfer with QSI of $\rho$ as follows.

\begin{Def} \label{def:additional_QSI_effect}
Let $\textbf{E}[O]_{i_1,j_1}^{i_2,j_2}$ be defined as
\begin{equation}
\textbf{E}[O]_{i_1,j_1}^{i_2,j_2}=\textbf{E}[O]_{i_2,j_2}-\textbf{E}[O]_{i_1,j_1}, \nonumber
\end{equation}
where $\textbf{E}[O]_{i,j}$ is the effect of QSI $A_{1}\cdots A_{i}$ and $B_{1}\cdots B_{j}$
on the optimal resource cost of type $O$
in the state transfer with QSI of $\rho$
as in Definition~\ref{def:effect_of_QSI}.
Then we call $\textbf{E}[O]_{i_1,j_1}^{i_2,j_2}$
the \emph{additional effect on the optimal resource cost of type $O$
with respect to QSI $A_{i_1+1}\cdots A_{i_2}$ and $B_{j_1+1}\cdots B_{j_2}$}
in the state transfer with QSI of $\rho$.
\end{Def}

Since Alice's QSI and Bob's QSI independently affect the optimal resource costs
as shown in Theorem~\ref{thm:effect_of_QSI},
the additional effect $\textbf{E}[O]_{i_1,j_1}^{i_2,j_2}$
on the optimal resource cost of type $O$ can be written
in the form
\begin{equation}
\textbf{E}[O]_{i_1,j_1}^{i_2,j_2}
=\textbf{A}[O]_{i_1}^{i_2}+\textbf{B}[O]_{j_1}^{j_2} \nonumber
\end{equation}
where $\textbf{A}[O]_{i_1}^{i_2}=\textbf{E}[O]_{i_1,0}^{i_2,0}$
and $\textbf{B}[O]_{j_1}^{j_2}=\textbf{E}[O]_{0,j_1}^{0,j_2}$.
In the above equation,
$\textbf{A}[O]_{i_1}^{i_2}$ ($\textbf{B}[O]_{j_1}^{j_2}$)
means the additional effect of Alice's QSI $A_{i_1+1}\cdots A_{i_2}$
(Bob's QSI $B_{j_1+1}\cdots B_{j_2}$).
This together with Theorem~\ref{thm:effect_of_QSI} gives us the following theorem
which explains the effects
of the more QSI $A_{i_1+1}\cdots A_{i_2}$ and $B_{j_1+1}\cdots B_{j_2}$
in the state transfer with QSI of $\rho$.

\begin{Thm} \label{thm:effect_of_additional_QSI}
In the state transfer with QSI of $\rho$,
the additional effects of Alice's QSI $A_{i_1+1}\cdots A_{i_2}$ are given
by $\textbf{A}[e]_{i_1}^{i_2}=0$ and
\begin{eqnarray}
& &\textbf{A}[c]_{i_1}^{i_2}=2\textbf{A}[Q]_{i_1}^{i_2}
 =-2\textbf{A}[E]_{i_1}^{i_2} \nonumber \\
&=&I(C_{\mathrm{A}};A_{i_1+1}\cdots A_{i_2}|A_{1}\cdots A_{i_1}). \nonumber
\end{eqnarray}
For Bob's QSI $B_{j_1+1}\cdots B_{j_2}$,
the additional effects are
\begin{eqnarray}
& &\textbf{B}[c]_{j_1}^{j_2}=\textbf{B}[e]_{j_1}^{j_2}
=2\textbf{B}[Q]_{j_1}^{j_2}=2\textbf{B}[E]_{j_1}^{j_2} \nonumber \\
&=&I(C_{\mathrm{A}};B_{j_1+1}\cdots B_{j_2}|B_{1}\cdots B_{j_1}). \nonumber
\end{eqnarray}
\end{Thm}

Remark that in Theorem~\ref{thm:effect_of_additional_QSI}
only the additional effect $\textbf{A}[E]_{i_1}^{i_2}$ is non-positive,
while the other additional effects are non-negative.
Moreover,
by comparing Theorem~\ref{thm:effect_of_QSI} and Theorem~\ref{thm:effect_of_additional_QSI},
it is verified that
the effect $\textbf{A}[O]_{i_1}$ ($\textbf{B}[O]_{j_1}$)
and the additional effect $\textbf{A}[O]_{i_1}^{i_2}$
($\textbf{B}[O]_{j_1}^{j_2}$) on the optimal resource cost of type $O$
can have the same sign,
since the QCMI is always non-negative~\cite{W13}.
This means that the use of more QSI $A_{i_1+1}\cdots A_{i_2}$ and $B_{j_1+1}\cdots B_{j_2}$
can enhance the effects of QSI $A_{1}\cdots A_{i_1}$ and $B_{1}\cdots B_{j_1}$
in the state transfer with QSI of $\rho$.

\section{New operational meanings of QMI and QCMI in terms of QSI}
\label{sec:New_operational_meanings}
In this section,
we present new operational meanings of the QMI,
the QCMI,
and the chain rule for the QMI.

From the effects of Alice's QSI provided in Theorem~\ref{thm:effect_of_QSI},
we can obtain
the following new operational meanings of the QMI,
which have never been considered before.

\begin{Cor}[Operational meanings of QMI] \label{cor:OM_QMI}
Let $\rho_{CS}$ be a quantum state.
Consider the state merging with QSI of $\rho_{CS}$, 
in which $C$ is merged from Alice to Bob.

(i) If Alice has $S$ and uses it as QSI,
then $I(C;S)$ can be interpreted as how much the classical communication cost can be reduced
compared to the case that Alice uses no QSI.

(ii) If Bob has $S$ and uses it as QSI,
then $I(C;S)$ can be interpreted
as how much both classical communication and entanglement costs can be reduced
compared to the case that Bob uses no QSI.
\end{Cor}

The additional effects of Alice's more QSI in Theorem~\ref{thm:effect_of_additional_QSI}
provides us new operational meanings of the QCMI,
which have never appeared in any previous literature.

\begin{Cor}[Operational meanings of QCMI] \label{cor:OM_QCMI}
Let $\rho_{CS_1S_2}$ be a quantum state.
Consider the state merging with QSI of $\rho_{CS_1S_2}$, in which $C$ is merged from Alice to Bob.

(i) If Alice has $S_1S_2$ and uses it as QSI,
then $I(C;S_2|S_1)$ means how much the classical communication cost can be more reduced
compared to the case that Alice uses QSI $S_1$ only.

(ii) If Bob has $S_1S_2$ and uses it as QSI,
then $I(C;S_2|S_1)$ means
how much both classical communication and entanglement costs can be more reduced
compared to the case that Bob uses QSI $S_1$ only.
\end{Cor}

We note that other operational meanings of the QMI and the QCMI
have been found in literature~\cite{HOW07,DY08}.
In both meanings,
one argument of the QMI and the QCMI is interpreted as the reference system.
This means that the operational meanings are explained
in terms of the reference system
which has nothing to do with the corresponding operational tasks.
On the other hand,
our operational meanings in Corollary~\ref{cor:OM_QMI} and Corollary~\ref{cor:OM_QCMI}
are intuitive and natural
since they only involve Alice's and Bob's systems without mentioning the reference.

In addition,
there is one more difference between our operational meanings and the others.
We first note that
each of the operational meanings
for the quantum conditional entropy~\cite{HOW07},
the QCMI~\cite{DY08},
and the min- and max-entropies~\cite{KRS09}
is obtained from one concrete operation.
However,
the state transfer with QSI can describe various operational situations
in which more (or less) QSI can be used.
From comparing these situations,
we can see that
the effects of QSI can be naturally derived,
and hence the QMI and the QCMI can be operationally interpreted
with respect to the effects,
even though each of them does not correspond to any concrete operation.

We furthermore remark that
if QSI $S_1S_2$ can be almost produced from QSI $S_1$
then the optimal cost of the state merging with QSI $S_1S_2$ is almost
the same as one of the state merging with QSI $S_1$ only.
Recently, it has been shown that
there is an important relation between the QCMI
and the recovery map through the Markov chain condition~\cite{FR15},
that is,
for any state $\rho=\rho_{CS_1S_2}$, 
there exists a quantum operation $R_{S_1\to S_1S_2}$ 
such that 
\begin{equation} \label{eq:FR15}
F(\rho,R_{S_1\to S_1S_2}\left(\rho_{CS_1}\right))\ge 2^{-\frac{1}{2}I(C;S_2|S_1)_\rho}.
\end{equation}
This implies that the converse of our above remark is also true.
Thus we can obtain the following corollary.
\begin{Cor} \label{Cor:Recover}
In the state merging with QSI $S_1S_2$,
the amount of the reduced cost 
by adding QSI $S_2$ to QSI $S_1$ is close to zero
if and only if
the QSI $S_1S_2$ can be almost recovered from the QSI $S_1$.
\end{Cor}
Moreover, the inequality~(\ref{eq:FR15}) also implies that
if the fidelity of its left-hand side decreases
then the QCMI $I(C;S_2|S_1)$ increases.
This means that
if QSI $S_1S_2$ cannot be properly recovered from QSI $S_1$
then the state merging with QSI $S_1S_2$ can have 
the more reduced optimal cost 
than that of the state merging with QSI $S_1$.

The chain rule~\cite{W13}
for the QMI is that
\begin{align}
&I(C;S_{1}\cdots S_{n}) \nonumber \\
&=I(C;S_{1})+I(C;S_{2}|S_{1})+\cdots+I(C;S_{n}|S_{1}\cdots S_{n-1}) \nonumber \\
&=I(C;S_{1}\cdots S_{i})+I(C;S_{i+1}\cdots S_{n}|S_{1}\cdots S_{i}) \label{eq:chain_rule}
\end{align}
for $1\le i\le n$,
where the first equality is the original chain rule 
but it can be simply rewritten by exploiting the rightmost side in Eq.~(\ref{eq:chain_rule}).  
From the concept of the state merging with QSI,
we can interpret the chain rule in Eq.~(\ref{eq:chain_rule}) as follows.
In the state merging with QSI,
the cost reduced by using the whole QSI $S_{1}\cdots S_{n}$ is equal 
to the sum of the cost reduced by using the partial QSI $S_{1}\cdots S_{i}$ 
and the cost more reduced by using the additional QSI $S_{i+1}\cdots S_{n}$.

\section{Examples of State transfer with QSI}
\label{sec:Examples_of_STwQSI}
\begin{figure}
\centering
\includegraphics[width=.9\linewidth,trim=0cm 0cm 0cm 0cm]{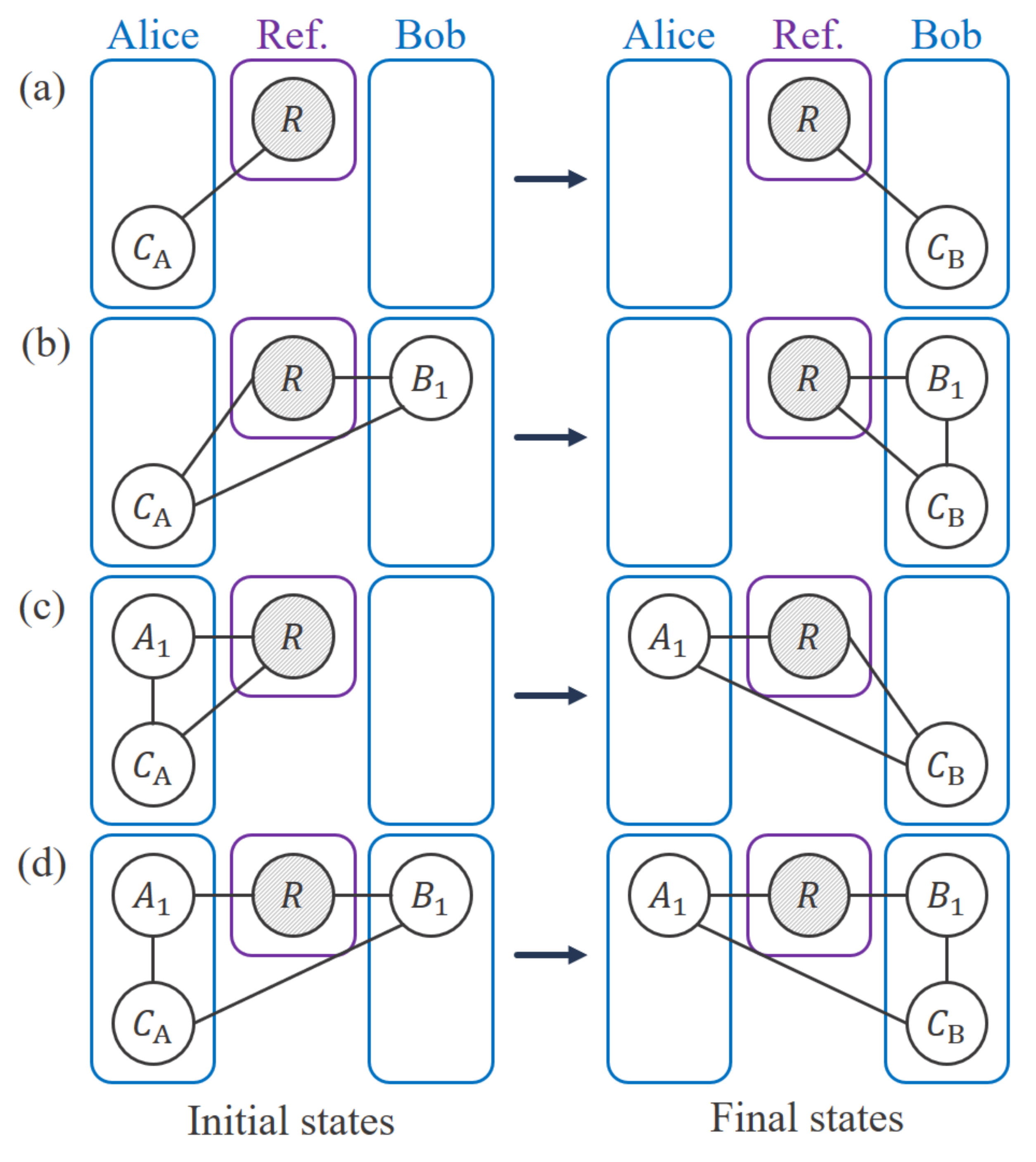}
\caption{
The initial and final states for protocols with qubit/bit channels
which can be classified into 
(a) SC/QT,
(b) FQSW/SM, 
(c) FQRS/GQT,
(d) SR/GSM,
according to the use of QSI. 
$C_{\mathrm{A}}$ is transferred from Alice to Bob,
$R$ is the reference, 
and $A_1$ and $B_1$ are Alice's and Bob's QSI, respectively.
}
\label{fig:Examples}
\end{figure}

Our protocol includes many well-known protocols of quantum information theory
in the sense that their optimal resource costs directly obtained
from Eqs.~(\ref{eq:Opt_Cost}) and~(\ref{eq:Opt_Cost_SM}).
We present four protocols which exploit qubit channels and other four protocols using bit channels.
Denote $Q$ and $E$ by the optimal quantum communication and entanglement costs.

(i) {\it Schumacher compression} (SC):
In the state redistribution with QSI
of $\ket{\psi}_{A_{1}\cdots A_{m}C_{\mathrm{A}}B_{1}\cdots B_{n}R}$,
if any QSI does not exist, that is, $m=n=0$, 
then the protocol becomes the SC~\cite{S95} as depicted in Fig.~\ref{fig:Examples}~(a).
From Eq.~(\ref{eq:Opt_Cost}),
we have $Q=H(C_{\mathrm{A}})$ and $E=0$, which are the optimal resource costs for SC.

(ii) {\it Fully quantum Slepian-Wolf} (FQSW):
FQSW~\cite{D06,ADHW09} is described in Fig.~\ref{fig:Examples}~(b),
which is a special case of our state redistribution with QSI 
if Alice does not have any QSI but Bob can use his QSI, that is, $m=0$ and $n=1$.
$Q=H(C_{\mathrm{A}})-\frac{1}{2}I(C_{\mathrm{A}};B_1)$ and $E=-\frac{1}{2}I(C_{\mathrm{A}};B_1)$ 
computed from Eq.~(\ref{eq:Opt_Cost})
are identical to the optimal costs of FQSW.

(iii) {\it Fully quantum reverse Shannon} (FQRS):
FQRS~\cite{D06,ADHW09} can be considered as the state redistribution with QSI
when Alice has QSI $A_1$ but Bob does not as in Fig.~\ref{fig:Examples}~(c),
that is,
$m=1$ and $n=0$.
Using Eq.~(\ref{eq:Opt_Cost}),
its optimal costs are given 
by $Q=H(C_{\mathrm{A}})-\frac{1}{2}I(C_{\mathrm{A}};A_1)$ and $E=\frac{1}{2}I(C_{\mathrm{A}};A_1)$,
which are equivalent to the optimal costs of FQRS.

(iv) {\it State redistribution} (SR):
In SR~\cite{DY08,YD09},
both Alice and Bob have QSI $A_1$ and $B_1$ as (d) in Fig.~\ref{fig:Examples},
that is, $m=1$ and $n=1$.
Its optimal resource costs 
$Q=H(C_{\mathrm{A}})-\frac{1}{2}I(C_{\mathrm{A}};A_1)-\frac{1}{2}I(C_{\mathrm{A}};B_1)$ 
and $E=\frac{1}{2}I(C_{\mathrm{A}};A_1)-\frac{1}{2}I(C_{\mathrm{A}};B_1)$
can be achieved
from Eq.~(\ref{eq:Opt_Cost}).

As mentioned earlier,
we continue to see the protocols with bit channels,
which are contained in the state merging with QSI
of $\ket{\psi}_{A_{1}\cdots A_{m}C_{\mathrm{A}}B_{1}\cdots B_{n}R}$.
Let us now define $c$ and $e$ as the optimal classical communication and entanglement costs, 
respectively.

(v) {\it Quantum teleportation} (QT):
In the original QT~\cite{BBCJPW93}, Alice and Bob can teleport only one qubit unknown to them.
However, 
we here assume that they asymptotically teleport an initial state known to themselves.
Then its optimal costs can be obtained as
$c=2H(C_{\mathrm{A}})$ and $e=H(C_{\mathrm{A}})$ from Eq.~(\ref{eq:Opt_Cost_SM}).
This is described in (a) of Fig.~\ref{fig:Examples}, as in the case of SC.

(vi) {\it State merging} (SM):
In SM~\cite{HOW05,HOW07},
Alice has no QSI but Bob has QSI, 
as depicted in (b) of Fig.~\ref{fig:Examples}.
This is equivalent to FQSW except for using different kind of channels.
From Eq.~(\ref{eq:Opt_Cost_SM}), its optimal costs $c=2H(C_{\mathrm{A}})-I(C_{\mathrm{A}};B_1)$
and $e=H(C_{\mathrm{A}})-I(C_{\mathrm{A}};B_1)$ can be obtained.

(vii) {\it Generalized quantum teleportation} (GQT) and \emph{Generalized state merging} (GSM):
In QT and SM,
if Alice has QSI $A_1$ and exploits it for teleporting and merging $C_{\mathrm{A}}$,
then we call these protocols GQT and QSM, 
which are seen in (c) and (d) of Fig.~\ref{fig:Examples},
respectively.
We note that the concepts of the GQT and the GSM have been known
in literature~\cite{F04,LKPL09,DY08,O08},
but the optimal resource costs have not precisely been mentioned.
By using Eq.~(\ref{eq:Opt_Cost_SM}),
it can be shown that $c=2H(C_{\mathrm{A}})-I(C_{\mathrm{A}};A_1)$ and $e=H(C_{\mathrm{A}})$
are the optimal costs for GQT,
and $c=2H(C_{\mathrm{A}})-I(C_{\mathrm{A}};A_1)-I(C_{\mathrm{A}};B_1)$
and $e=H(C_{\mathrm{A}})-I(C_{\mathrm{A}};B_1)$ for GSM.

So far,
we have seen that the state transfer with QSI includes 
many quantum information protocols to transfer Alice's information to Bob,
and our protocol is the most generalized one when taking account of Alice's and Bob's QSI.

\section{Conclusion} \label{sec:Conclusion}
We have considered the state transfer with QSI
as a general quantum communication protocol,
and have determined its optimal resource costs
when Alice and Bob use their QSI.
We also have investigated the effects of (additional) QSI
on the optimal resource costs
in the state transfer with QSI.
Based on this study,
we have provided new operational meanings 
of the QMI 
and the QCMI,
in addition to a new operational interpretation 
of the chain rule for the QMI,
which is naturally understandable
with respect to the state transfer with QSI.
In addition,
we expect that our state transfer with QSI provides 
further understandings of specific multipartite quantum states,
such as the Greenberger-Horne-Zeilinger state~\cite{GHZ89} and the Werner state~\cite{W89}.

Throughout this paper,
we have assumed that the initial states of the protocols
are independent and identically distributed (i.i.d.).
However,
there have been some results~\cite{B08,BCR11,RR12,DH13,BCT16}
which do not take into account the i.i.d. assumption.
Since these results have provided theoretical bases
for the proofs of some practical applications,
such as quantum key distribution with finite resources~\cite{TLGR12,CXCLTL14},
it can be helpful to devise 
the one-shot version of our work.
For this, recent results about resource costs
for the one-shot quantum state redistribution~\cite{BCT16,AJW17}
might be useful.

Furthermore,
it would be interesting to investigate 
the optimal resource costs of the state transfer with QSI 
under various conditions.
For instance,
we can assume that 
Alice and Bob can consume non-local noisy resources~\cite{DHW04,DHW08},
or they can use a local resource,
such as maximally coherent states~\cite{BCP14,SSDBA15,SCRBWL16},
as in the incoherent quantum state merging~\cite{SCRBWL16}
and the coherence distillation~\cite{WY16}.

We thank Alexander Streltsov for very helpful comments.
This research was supported by Basic Science Research Program 
through the National Research Foundation of Korea (NRF) funded 
by the Ministry of Science and ICT (NRF-2016R1A2B4014928).

\bibliography{STwQSI}
\end{document}